\documentclass[10pt]{article}
\usepackage{amsmath,amssymb,mathtools}
\usepackage{soul}
\usepackage{graphicx,color}
\usepackage{xspace}
\usepackage{todonotes}
\usepackage{csquotes}
\usepackage{booktabs} 
\usepackage{tabularx} 
\usepackage{hyperref}
\usepackage{thm-restate}
\newcommand{\lv}[1]{}

\usepackage{tikz,authblk}

\usepackage[small]{complexity}
\usepackage[T1]{fontenc} 

\usepackage{microtype}

\usepackage{hyphenat}

\usepackage{xpunctuate}
\usepackage[all,british]{foreign} 
\redefnotforeign[ie]{i.e\xperiodafter}
\redefnotforeign[eg]{e.g\xperiodafter}

\usepackage{tikz}

\def\dist{\operatorname{d}}

\usepackage{stmaryrd}

\usepackage{accents}
\newlength{\dhatheight}

\renewcommand{\le}{\leqslant}
\renewcommand{\leq}{\leqslant}
\renewcommand{\ge}{\geqslant}
\renewcommand{\geq}{\geqslant}

\renewcommand{\epsilon}{\varepsilon}

\def\pruned{\widetilde}
\def\symdiff{\mathop{\vartriangle}}

\newenvironment{tightcenter}
 {\parskip=0pt\par\nopagebreak\centering}
 {\par\noindent\ignorespacesafterend}

\usepackage{tikz}
\usetikzlibrary{calc}
\usepackage{xspace}
\usepackage{xargs}
\usepackage{xifthen}

\usepackage{framed}

\usepackage{ctable}
\newlength{\RoundedBoxWidth}
\newsavebox{\GrayRoundedBox}
\newenvironment{GrayBox}[1]%
   {\setlength{\RoundedBoxWidth}{\textwidth-4.5ex}
    \def\boxheading{#1}
    \begin{lrbox}{\GrayRoundedBox}
       \begin{minipage}{\RoundedBoxWidth}%
   }{%
       \end{minipage}
    \end{lrbox}%
    \begin{tightcenter}%
    \begin{tikzpicture}%
       \node(Text)[draw=black!20,fill=white,rounded corners,%
             inner sep=2ex,text width=\RoundedBoxWidth]%
             {\usebox{\GrayRoundedBox}};
        \coordinate(x) at (current bounding box.north west);
        \node [draw=white,rectangle,inner sep=3pt,anchor=north west,fill=white] 
        at ($(x)+(6pt,.75em)$) {\boxheading};
    \end{tikzpicture}
    \end{tightcenter}\vspace{0pt}%
    \ignorespacesafterend
}    

\newenvironment{problem}[2][]{\noindent\ignorespaces%
                                \FrameSep=6pt%
                                \parindent=0pt%
                \vspace*{-.5em}
                \ifthenelse{\isempty{#1}}{%
                  \begin{GrayBox}{\textsc{#2}}%
                }{%
                  \begin{GrayBox}{\textsc{#2} parameterized by~{#1}}%
                }
                \newcommand\Prob{Problem:}%
                \newcommand\Input{Input:}%
                \begin{tabular*}{\textwidth}{@{\hspace{.1em}} >{\itshape} p{1.6cm} p{0.8\textwidth} @{}}%
            }{
                \end{tabular*}%
                \end{GrayBox}%
                \vspace*{-.5em}
                \ignorespacesafterend
            }

\newcommand{\Problem}[1]{\textsc{#1}}

\usepackage[amsmath,amsthm,thmmarks]{ntheorem}
\theoremseparator{.}
\newtheorem{theorem}{Theorem}
\newtheorem{observation}[theorem]{Observation}
\newtheorem{lemma}[theorem]{Lemma}
\newtheorem{corollary}[theorem]{Corollary}
\newtheorem{proposition}[theorem]{Proposition}
\newtheorem{definition}[theorem]{Definition}

\newcommand{\md}{\textrm{md}}
\newcommand{\mdim}{{\sc Metric Dimension}}

\newcommand{\qedhere}{\ifmmode\qed\else\hfill\proofSymbol\fi}

\def\Wahlstrom{Wahlstr{\"o}m\xspace}

\usepackage{algorithm}
\usepackage{algorithmic}

\newcommand{\cF}{\ensuremath{\mathcal{F}}}

\def\plog{\log^{\kern-.1pt{\scriptscriptstyle O(1)}}\kern-2pt}

\begin{document}

\title{Alternative parameterizations of \Problem{Metric~Dimension}}

\author[1]{Gregory Gutin}
\author[2]{M. S. Ramanujan}
\author[1]{Felix Reidl}
\author[1]{Magnus \Wahlstrom}
\affil[1]{Royal Holloway, University of London, TW20 0EX, UK}
\affil[2]{University of Warwick, CV4 7AL, UK}
\date{}
\maketitle

\begin{abstract}  
A set of vertices $W$ in a graph $G$ is called \emph{resolving} if
for any two distinct $x,y\in V(G)$, there is $v\in W$ such that
$\dist_G(v,x)\neq\dist_G(v,y)$, where $\dist_G(u,v)$ denotes the length of a
shortest path between $u$ and $v$ in the graph $G$. The metric dimension
$\md(G)$ of $G$ is the minimum cardinality of a resolving set. The \Problem{Metric
Dimension} problem, \ie deciding whether $\md(G)\le k$, is NP-complete
even for interval graphs (Foucaud \etal, 2017). We study \Problem{Metric
Dimension} (for arbitrary graphs) from the lens of parameterized
complexity. The problem parameterized by $k$ was proved to be $\W[2]$-hard by
Hartung and Nichterlein (2013) and we study the dual parameterization, i.e.,
the problem of whether $\md(G)\le n- k,$ where $n$ is the order of $G$. We
prove that the dual parameterization admits (a) a kernel with at most $3k^4$
vertices and (b) an algorithm of runtime $O^*(4^{k+o(k)}).$ Hartung and
Nichterlein (2013) also observed that \Problem{Metric Dimension} is fixed-parameter
tractable when parameterized by the vertex cover number $vc(G)$ of the input graph.
We complement this observation by showing that it does not admit a polynomial kernel
even when parameterized by~$vc(G) + k$.
Our reduction also gives evidence for non-existence of polynomial Turing kernels.  
\end{abstract}

\section{Introduction} 
\label{sec:intro}

A set of vertices $W$ of a graph $G$ is a \emph{resolving} set for $G$ if for
any two distinct $x,y\in V(G)$, there is $v\in W$ such that
$\dist_G(v,x)\neq\dist_G(v,y)$, where $\dist_G(u,v)$ denotes the length of a
shortest path between $u$ and $v$ in the graph $G$. The \emph{metric
dimension} $\md(G)$ of $G$ is the minimum cardinality of a resolving set for
$G$. The metric dimension of graphs was introduced independently by
Slater~\cite{Slater75} and Harary and Melter~\cite{HararyM76}. \Problem{Metric Dimension}
as a computational problem was first mentioned in the literature by
Garey and Johnson~\cite{GareyJ79} and its decision version is defined as
follows.
 
\begin{problem}{Metric Dimension}
  \Input & A graph~$G$ and an integer~$k$. \\
  \Prob  & Does~$G$ have a resolving set of size at most~$k$? 
\end{problem}

\noindent Garey and Johnson~\cite{GareyJ79} proved this problem to be 
NP-complete in general. Their proof was never published, a reduction
from 3SAT was provided by Khuller \etal~\cite{KhullerRR96}.
Diaz \etal~\cite{DiazPSL12} showed that the problem is NP-complete even when
restricted to planar graphs of bounded degree but that it is
solvable in polynomial time on the class of outer-planar graphs.

Prior to this, not much was known about the computational complexity of this
problem except that it is polynomial-time solvable on trees
(see~\cite{Slater75,KhullerRR96}), although there are several results proving
combinatorial bounds on the metric dimension of various graph
classes~\cite{ChartrandEJO00}.  Subsequently, Epstein et
al.~\cite{EpsteinLW12} showed that this problem is NP-complete on split
graphs, bipartite and co-bipartite graphs. They also showed that the
\emph{weighted version} of {\mdim} can be solved in polynomial time on paths,
trees, cycles, co-graphs and trees augmented with $k$ edges for a fixed $k$.
Hoffmann and Wanke~\cite{HoffmannW12} extended the tractability results to a
subclass of unit disk graphs, while Foucaud \etal~\cite{FoucaudMNPV17} showed
that this problem is NP-complete on interval graphs.

The parameterized complexity of {\mdim} under the standard  
parameterization---the metric dimension of the input graph---was open until 2012, when
Hartung and Nichterlein~\cite{HartungN13} proved that it is $\W[2]$-hard.
Foucaud \etal~\cite{FoucaudMNPV17} showed the problem becomes fixed-parameter
tractable when restricted to interval graphs.   The parameterized complexity
of \Problem{Metric Dimension} on graphs of bounded treewidth is currently
unresolved (the question of whether it is polynomial-time solvable on graphs
of treewidth 2 is still open), however, Belmonte \etal~\cite{BelmonteFGR17}
proved that it is $\FPT$ when parameterized by the
tree\emph{length}\footnote{The \emph{length} of a tree decomposition is the
maximum diameter of the bags in this tree-decomposition and the
\emph{treelength} of a graph is the minimum length over all tree
decompositions.  Note that this parameter is upper-bounded by treewidth.} plus
the solution size. In a different line of work,
Eppstein~\cite{Eppstein15} showed that \Problem{Metric Dimension} is $\FPT$ when
parameterized by the max-leaf number of the input graph alone. 

In this paper we initiate the study of the parametric dual of
\Problem{Metric Dimension}. To avoid confusion, we will use $k$ to denote the
(standard) parameter and phrase the parameterized dual as follows:

\begin{problem}{Saving Landmarks}
  \Input & A graph~$G$ and an integer~$k$. \\
  \Prob  & Does~$G$ have a resolving set of size at most~$n-k$? 
\end{problem}

\noindent
We call a set $T$ of vertices of $G$ a {\em co-resolving set} if
$V(G)\setminus T$ is a resolving set of $G$. Clearly, an instance of
\Problem{Saving Landmarks} is positive if and only if there is a co-resolving
set $T$ of size at least $k$.

This choice of parameterization is informed by previous studies of the parametric
dual (see e.g. \cite{BasavarajuFRS16,CrowstonGJSY12,GutinJY11,RazgonO09}):
problems that are hard with respect to the standard parameter often admit an
$\FPT$-algorithms or even polynomial kernels under the dual parameter.
A classic example is the
\Problem{Independent Set} problem which is $\W[1]$-hard while its dual, the
\Problem{Vertex Cover} problem is among the earliest problems shown to be in
$\FPT$ and even admits a linear vertex kernel.  

We add yet another entry to the list of hard problems with tractable duals by
showing that \Problem{Saving Landmarks} admits a polynomial kernel and a
single-exponential $\FPT$ algorithm. Concretely, we prove the following two
results.

\begin{restatable}{theorem}{savingkernel}\label{thm:saving-kernel}
  \Problem{Saving Landmarks} admits a kernel with at most~$8k^4$ vertices.
\end{restatable}

\begin{restatable}{theorem}{savingalgo}\label{thm:saving-algo}
  \Problem{Saving Landmarks} can be solved in time~$O^*(4^{k+o(k)})$.
\end{restatable}

\noindent  
We also study the \Problem{Metric Dimension} problem from the
kernelization perspective when parameterized by the vertex cover number of the
input graph.  As Hartung and Nichterlein observed~\cite{HartungN13},
parameterization of \Problem{Metric Dimension}  by the vertex cover number of
the input graph  (denoted \Problem{Metric Dimension[VC]}) can be easily seen
to be in $\FPT$. It is therefore natural to ask whether this structural
parameterization allows a polynomial kernel in general graphs, a question we
answer in the negative. In fact, we show that not only does the problem not
admit a polynomial kernel with the vertex cover as the parameter, even adding
the size of the solution (the metric dimension of the graph) to the parameter
is unlikely to be helpful in this regard. Specifically, we prove the following
result.

\begin{restatable}{theorem}{kernelhardness}\label{nopolykernel}
  \Problem{Metric Dimension[$\text{VC}+k$]} does not admit
  a polynomial kernel unless the polynomial hierarchy collapses
  to its third level. 
\end{restatable}

\noindent 
The reduction used in the proof of Theorem \ref{nopolykernel} also
gives evidence for non-existence of polynomial Turing kernels,
generalizations of (ordinary) kernels, informally introduced in the end of
Section \ref{sec:nopolykernel}.

\section{Preliminaries}\label{sec:prelims}

For a graph~$G$ we denote by~$\dist_G$ the standard distance-metric
where~$\dist_G(u,v)$ is the length of a shortest path between vertices $u,v
\in V(G)$. We denote by~$N_G(v)$ and~$N_G[v]$ the open and closed
neighbourhood of a vertex. We omit the subscript~$G$ if clear from the context
in all these notations. As customary, the number of vertices of a graph $G$
under consideration will be denoted by $n.$

Two vertices~$u,v$ are \emph{true
twins} if~$N_G[u] = N_G[v]$ (implying that~$uv \in G$) and they
are \emph{false twins} if~$N_G(u) = N_G(v)$. A \emph{twin class}
is a maximal vertex set in~$G$ in which all vertices are pairwise
true twins or in which all vertices are pairwise false twins. 

A vertex set~$S \subseteq V(G)$ \emph{resolves} a set~$T \subseteq V(G)$ if 
for every pair of distinct vertices~$u,v \in T$ there exists at least
one vertex~$w \in S$ such that~$\dist_G(u,w) \neq \dist_G(v,w)$.
We will also say that a pair~$u,v$ is {\em resolved} by~$S$ if the above holds
and further that sets~$A,B$ are \emph{distinguished} by~$S$ if every pair
$u \in A$, $v \in B$ is resolved by~$S$.
A vertex subset~$S \subseteq V(G)$ is a \emph{resolving set} of~$G$
if~$S$ resolves~$V(G)$. We call the members of such a set~$S$ \emph{landmarks}. 

Parameterized complexity is a two dimensional framework for
studying the computational complexity of a problem. One dimension is the input
size $n$ and the other is a parameter $k$. A problem is said to be \emph{fixed
parameter tractable} ($\FPT$) or in the class $\FPT$, if it can be
solved in time $f(k)\cdot n^{O(1)}$ for some computable function $f$. We refer
to the books of  Cygan \etal~\cite{CyganFKLMPPS15} and Downey and
Fellows~\cite{DowneyF13} for  detailed introductions  to parameterized
complexity.

Kernelization offers a mathematically rigorous way  of analysing and
comparing preprocessing algorithms for $\NP$-hard problems in general and for
parameterized problems in particular.  A {\em kernel of size $g(k)$} for a
parameterized problem  is a polynomial time algorithm that takes as input an
instance $(I,k)$ of the problem (where $k$ is the parameter) and outputs
another instance $(I',k')$ of the same problem such that $(I,k)$ is a yes-
instance of the problem if and only if $(I',k')$ is a yes-instance of the
problem and $|I'|+k' \leq g(k)$. The notion of ``effective'' preprocessing is
captured by requiring the function $g$ to be polynomially bounded, in which
case the kernel is called a {\em polynomial kernel}. The reader is referred to
Cygan \etal~\cite{CyganFKLMPPS15}, Downey and Fellows~\cite{DowneyF13}, Fomin
\etal \cite{FominLSZ} and the surveys~\cite{Kratsch14,LokshtanovMS12}  for  a
comprehensive introduction to the topic of kernelization.

\begin{definition}[Pruned graph]
  For a graph~$G$ we define the \emph{pruned graph}~$\pruned G$ as
  the graph obtained (up to isomorphism) from~$G$ by iteratively removing
  vertices from twin-classes of size three or larger. We say that
  a graph is \emph{pruned} if~$G = \pruned G$.
\end{definition}

\noindent 
The following observation simply follows from the fact that
among a twin class~$U$ in~$G$, all but one vertex of~$U$
must be contained in any resolving set.

\begin{observation}\label{obs:pruning}
  A graph~$G$ has a resolving set of size~$k$ if and only if
  the pruned graph $\tilde G$ has a resolving set of size~$k - (|V(G)| - |V(\tilde G)|)$.
\end{observation}

\noindent 
Consequently, we call an instance~$(G,k)$ of \Problem{Metric
Dimension} or \Problem{Saving Landmarks} \emph{reduced} if $G$ is pruned.

%
%
%
\section{Standard parameterization for \Problem{Saving Landmarks}}

We present two positive results in this section, namely, that
\Problem{Saving Landmarks} admits a polynomial kernel and a 
single-exponential $\FPT$ algorithm.

%

%
%

We begin by describing the kernel. Assume in the following that the input
instance~$(G,k)$ is pruned as per Observation~\ref{obs:pruning}. This will be
the only reduction rule. In the following we will prove that the size of the
instance is either bounded polynomially in~$k$ or it will be a trivial yes-instance. Let us collect some basic observations first.

\begin{lemma}\label{lemma:clique-indset}
  If~$G$ contains either an independent set or a clique with~$2k$ vertices,
  then~$(G,k)$ is a yes-instance.
\end{lemma}
\begin{proof}
  Let~$X$ be a set of size~$2k$ such that~$G[X]$ is either a clique or
  an independent set. Since~$G$ is pruned there are at least~$|X|/2 \geq k$
  distinct twin-classes in~$X$, which must be distinguished by their
  neighbourhoods outside of~$X$. Hence selecting one vertex from each
  twin-class in~$X$ gives a co-resolving set of size at least~$k$, 
  and~$(G,k)$ is a yes-instance.
\end{proof}

\noindent
Let us define the function~$\tau(u,v) := |N(u) \symdiff N(v)|$. Note that
if~$\tau(u,v) \geq k+1$, then~$u$ and~$v$ are distinguished from each other
by any set of $n-k$~landmarks, simply by virtue of having a necessarily
different set of landmarks as neighbours. Let us therefore construct
an auxiliary graph~$H$ on~$V(G)$ where
\[
  uv \in E(H) \iff \tau(u,v) \leq k.
\]
Observe that if~$H$ contains an independent set~$X$ of size~$k$ then~$(G,k)$
is a yes-instance: the set~$V(G)\setminus X$ has size~$n-k$ and as such will
still resolve all of~$X$. This indicates that~$H$ must be rather dense, however,
we can also argue that it cannot have arbitrarily high degree:

\begin{lemma}\label{lemma:H-degree-bound}
  Let~$v \in V(H)$ have degree at least~$8k^3$ in~$H$. Then~$(G,k)$ is 
  a yes-instance.
\end{lemma}
\begin{proof}
  Let~$S := N_H[v]$. Note that for every pair~$u,u' \in S$ it holds that
  \[
    \tau(u, u') \leq \tau(u,v) + \tau(u',v) \leq 2k.
  \]
  Now turn our attention to $G$. 
  First consider the case in which every vertex in~$G[S]$ has degree less
  than~$4k^2$. Then greedily packing closed neighbourhoods gives an
  independent set in~$G$ of size at least~$2k$, and by
  Lemma~\ref{lemma:clique-indset}, $(G,k)$ is a yes-instance. 

  Thus consider the alternative that~$G[S]$ contains a vertex~$u_1$ of degree
  at least~$4k^2$. Define~$S_1 := N_{G[S]}(u_1)$ and
  pick any vertex~$u_2 \in S_1$. Note that since~$\tau(u_1,u_2) \leq 2k$ it follows that
  \[
    |N_{G[S]}(u_1) \symdiff N_{G[S]}(u_2)| \leq |N_G(u_1) \symdiff N_G(u_2) | \leq 2k.
  \]
  Consequently, $u_1$ and~$u_2$ share at least~$|S_1| - 2k -1 \geq
  4k^2 - 2k - 1$ neighbours in~$G[S]$ (removing one extra since
  $u_2 \in S_1$). We can repeat this procedure to construct a sequence
  of distinct vertices~$u_1,u_2,\ldots,u_r$ and subsets~$S_1 \supseteq
  S_2 \supseteq \ldots \supseteq S_r$ 
  where~$S_i := \bigcap_{j \leq i} N_{G[S]}(u_j)$ and~$u_i \in
  S_{i-1}$ is chosen arbitrarily. The sequence terminates 
  with~$S_r = \emptyset$, giving a clique in $G$ of size $r$.
  Since $|S_i| \geq |S_{i-1}|-2k-1$ for every $i \in [r]$, 
  we get $|S_{2k-1}| \geq |S_1| - (2k-2)(2k+1) > 0$ since
  $|S_1| \geq 4k^2$.  Thus $r \geq 2k$ and $u_1$, \ldots, $u_r$ 
  induces a clique of size at least~$2k$ in $G$,
  and again by Lemma~\ref{lemma:clique-indset} we conclude
  that~$(G,k)$ is a yes-instance. 
\end{proof}

\noindent
With these pieces in place, we can prove the first result of this section.

\savingkernel*
\begin{proof}
  By Lemma~\ref{lemma:H-degree-bound} we either have  that $(G,k)$ is a yes-instance 
  or that the auxiliary graph~$H$ has a maximum degree less than~$8k^3$. Assuming
  the latter, if~$|V(G)| = |V(H)| \geq 8k^4$ then~$H$ contains an independent
  set of size at least~$k$ and, as observed above, $(G,k)$ is a yes-instance.

  The kernel for \Problem{Saving Landmarks} is therefore the following procedure:
  for a given instance~$(G',k')$, compute the reduced instance~$(G,k)$. If~$G$
  contains more than~$8k^4$ vertices, return a trivial yes-instance. Otherwise,
  return~$(G,k)$. 
\end{proof}

%
\noindent
Let us now move on to the second result, the single-exponential $\FPT$ algorithm.
To better describe the algorithm, let us introduce a definition.  For
a set $X \subseteq V(G)$, we say that two vertices $u$ and $v$ are
\emph{$X$-equidistant} if $\dist(u,w)=\dist(v,w)$ for every $w \in X$,
i.e., if $X$ fails to resolve $u$ and $v$.  Note that this induces an
equivalence relation over $V(G)$.

The main ingredient will be fact that a solution to \Problem{Saving Landmarks}
is witnessed already by a small resolving set.

\begin{lemma}\label{lemma:resolve-witness}
  Let~$T$ be a co-resolving set of a graph $G$. Then there exists a set~$S \subseteq V(G)\setminus T$ 
  of size at most~$|T|$ that resolves~$T$.
\end{lemma}
\begin{proof}
  We construct~$S$ iteratively as follows. Begin with~$S = \emptyset$ and pick
  a pair~$u,v$ of $S$-equidistant vertices in~$T$. Since~$V(G)\setminus T$ resolves~$T$,
  there exists a vertex~$w \in V(G)\setminus T$ that distinguishes~$u$ and~$v$. Add~$w$
  to~$S$ and partition~$T$ into equivalence classes of $S$-equidistant
  vertices. Pick a
  new pair of $S$-equidistant vertices from one of the classes and repeat. Observe that the
  number of equivalence classes increase with every addition to~$S$, hence after at most~$|T|$
  steps the set~$S$ resolves every pair in~$T$.
\end{proof}

\noindent
We are now ready to complete the proof of Theorem~\ref{thm:saving-algo}.

\savingalgo*

\begin{proof} We may assume that $n\ge 2k$. Let us first show the following
claim: there exists a co-resolving set $T$ of $G$ of size at least $k$ if and
only if there is a partition $V(G) = R \cup B$  of $V(G)$  such that $R$
contains at least $k$ equivalence classes of  $B$-equidistant vertices.
Suppose that there exists a co-resolving set $T$ of $G$ of size at least $k$.
Then by Lemma \ref{lemma:resolve-witness}, there is a set $S\subseteq
V(G)\setminus T$  of size at most~$|T|$ that resolves~$T$. Let $T\subseteq R$
and $S\subseteq B$ for a partition  $V(G)=R \cup B$. Then $B$ resolves $T$ and
hence $R$ has at least $|T|\ge k$ equivalence classes of $B$-equidistant
vertices. Suppose now that there is a partition $R \cup B$ of $V(G)$  such
that $R$ has at least $k$ equivalence classes of $B$-equidistant vertices.
Choose a vertex from each equivalence class to form a set $T$. Then $T$ is a
co-resolving set of $G$.

The above claim leads to the following randomized algorithm. Choose a natural
number $N$ defined later on. Repeat $N$ times the following: uniformly at
random partition the vertices of~$G$ into $B$ and $R$, and derive equivalence
classes of $B$-equidistant vertices in $R$. If the number of classes is at
least $k$, then conclude  that~$(G,k)$ is a yes-instance and stop. If after
all repetitions we do not conclude that $(G,k)$ is a yes-instance,  then we
conclude that $(G,k)$ is a no-instance.

Let us argue about the success probability of the randomized algorithm and how
to choose $N$. The probability that for a random partition  the vertices
of~$G$ as $V(G)=R \cup B$,  $R$ has at least $k$ equivalence classes of
$B$-equidistant vertices is at least the probability that $T\subseteq R$ and
$S\subseteq B$,  where sets $T,S$ are as in Lemma \ref{lemma:resolve-witness},
which is~$2^{-|T|-|S|} \ge  4^{-k}.$  Thus, $N=4^k$ is enough to achieve a
constant success probability~\cite{CyganFKLMPPS15}.

Observe that every loop in the randomized algorithm can be executed in
polynomial time. Thus, the running time of the randomized algorithm is
$O^*(4^k)$.  The randomized algorithm can be derandomized using the standard
$(n,k)$-universal set technique~\cite{CyganFKLMPPS15}, which brings an
additional $o(k)$ to the exponent of the running time. 
\end{proof}

%
%
\section{Structural parameterizations for \Problem{Metric Dimension}}\label{sec:nopolykernel}

As Hartung and Nichterlein observed~\cite{HartungN13}, \Problem{Metric
Dimension[VC]} is trivially  $\FPT$~by virtue of
Observation~\ref{obs:pruning}: After reducing the size of each twin class to
at most two, any instance with a vertex cover~$X$ of size~$t$ will have
at most~$t + 2^{t+1}$ vertices. In sparse graph classes, the twin reduction
even results in a polynomial-size kernel: in classes of bounded expansion (\eg planar graphs
or graphs excluding a topological minor), the number of twin classes
in~$V(G)\setminus X$ is bounded linearly in~$t$ and in nowhere dense classes
by~$t^{1+o(1)}$~(\cf Lemma~4.3 and Corollary~4.4 in~\cite{SparseKernels}).
Furthermore, if the input graphs stem from a~$d$-degenerate class, the number
of twin-classes and thus the number of vertices in the kernel is bounded by~$O(t^{d+1})$; a fact
that follows easily from the observation that in such a class at most~$dt$
vertices in the independent set can have degree more than~$d$.

It is therefore natural to ask whether this structural parameterization allows
a polynomial kernel in general graphs, a question we answer in the negative.
We will use in the following that \Problem{Hitting Set} parameterized by the
size of the universe plus the solution size does not admit a polynomial kernel
unless the polynomial hierarchy collapses to the third level
~\cite{ColsAndIds}

\kernelhardness*
\begin{proof}
  We provide a polynomial parameter transformation from~\Problem{Hitting
  Set[$|U|+\ell$]}, i.e. parameterized by the size of the universe and the
  solution size, to~\Problem{Metric Dimension[$\text{VC}+k$]}. 
  Let~$(U,\mathcal F, \ell)$ be a \Problem{Hitting Set} instance with~$n =
  |U|$ and~$m = |\mathcal F|$. We construct a graph~$G$ as follows (\cf Figure~\ref{fig:reduction}):
  \begin{enumerate}
    \item Begin with the usual bipartite representation of~$U,\mathcal F$,
      i.e., create a bipartite graph $G=(U \cup \cF, E)$ where
      for vertices $u \in U$ and $R \in \cF$ we have $uR \in E$ if and
      only if $u \in R$;
    \item add $t_n:=2 \lceil \log_2 n \rceil$ vertices~$I_U$ to the
      graph and edges between~$U, I_U$ so that every vertex in~$U$ has
      a unique neighbourhood in~$I_U$ of size~$t_n/2$;
    \item add~$t_m:= 2 \lceil \log_2 m \rceil$ vertices~$I_{\mathcal F}$ 
      to the graph and edges between~${\mathcal F}$ and $I_{\mathcal F}$ 
      such that every vertex in~${\mathcal F}$ has a unique neighbourhood in~$I_{\mathcal F}$ of size~$t_m/2$;
    \item add three vertices~$a_U, a, a_{\mathcal F}$
      where~$N(a_U)=U$, $N(a)=U \cup \mathcal F$,
      and $N(a_{\mathcal F})=\mathcal F$;
    \item create true twin copies~$I'_U, I'_{\mathcal F}, a'_U, a', a'_{\mathcal F}$ of
          $I_U, I_{\mathcal F}, a_U, a, a_{\mathcal F}$, and finally
    \item create false twin copies~$\mathcal F'$ of~$\mathcal F$ but remove all edges from~$\mathcal F'$
          to~$U$ afterwards. For simplicity, we will label the copy of any vertex~$R \in \mathcal F$ by~$R' \in \mathcal F'$.
  \end{enumerate}

  \begin{figure}[t]
    \hspace*{2pt}\includegraphics[width=\textwidth]{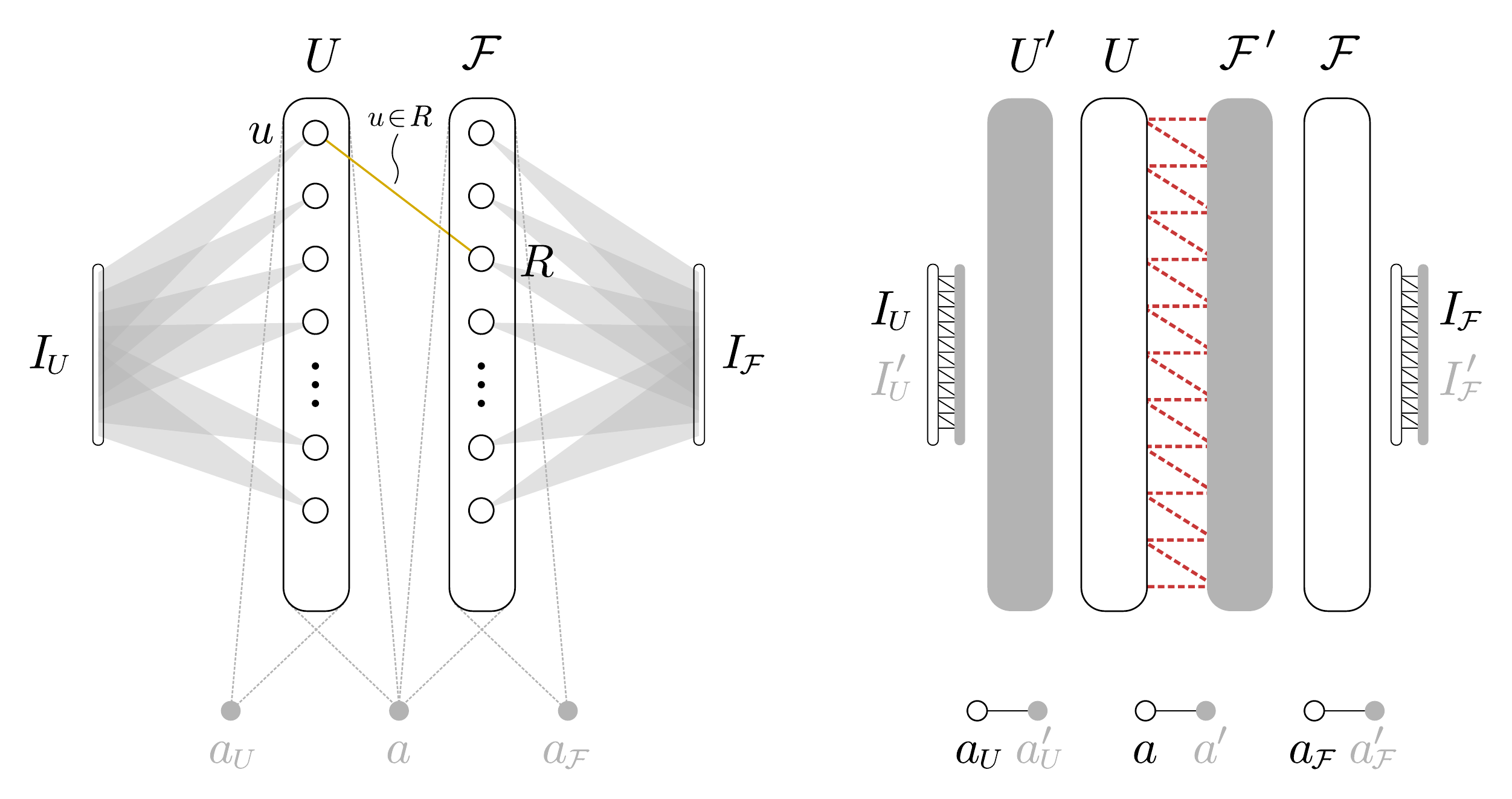}
    \caption{\label{fig:reduction}%
      A schematic of the reduction from a \Problem{Hitting Set[$|U|+\ell$]}
      instance~$(U,\mathcal F)$ to a \Problem{Metric Dimension[$\text{VC}+k$]}
      instance. The left drawing shows the basic construction, the right drawing
      the addition of false and true twins (an edge between a white set and its grey
      counterpart indicates that they are true twins, the absence of an edge that they
      are false twins). Note that the construction removes edges between the set
      $U$ and~$\mathcal F'$.
    }
  \end{figure}

  \noindent
  In summary, the sets~$I_U, I'_U$ connect to~$U$ only, the sets~$I_{\mathcal
  F}, I'_{\mathcal F}$ to~$\mathcal F$ and~$\mathcal F'$, the edges
  between~$U,\mathcal F$ encode the hitting set instance and the
  pairs~$\{a_U,a'_U\}$, $\{a,a'\}$, and~$\{a_\mathcal F,a'_\mathcal F\}$ are
  apices for the sets~$U$, $U \cup \mathcal F \cup \mathcal F'$ and~$\mathcal
  F \cup \mathcal F'$, respectively. Our construction concludes with~$(G,X,k)$
  as the \Problem{Metric Dimension[$\text{VC}+k$]} instance with the vertex cover~$X :=
  V(G) \setminus (\mathcal F \cup \mathcal F')$
  and solution size~$k := \ell + t_U + t_{\mathcal F} + 3$.

  Let us first show that if~$(U,\mathcal F,\ell)$ is a yes-instance
  then so is~$(G,X,k)$. Suppose that~$H \subseteq U$ is a hitting set
  for~$\mathcal F$ of size~$\ell$. We construct a landmark set~$S$ for~$G$ by
  setting $S = H \cup I_U \cup I_{\mathcal F} \cup \{a_U,a,a_{\mathcal F}\}$; let
  us now argue that is indeed a resolving set. First, note that the selected
  apices~$a_U$, $a$, and~$a_{\mathcal F}$ make sure that~$U$ is distinguished
  from~$V(G)\setminus U$ and~$\mathcal F \cup \mathcal F'$ from~$V(G)
  \setminus (\mathcal F \cup \mathcal F')$. Since~$I_U$ and~$I_{\mathcal F}$
  are in~$S$, these sets are of course distinguished from their twin
  counterparts~$I'_U, I'_{\mathcal F}$. By construction, every vertex in~$U$
  has a unique neighbourhood in~$I_U$, hence all of~$U$ is resolved by $S$.
  The same holds true for all pairs~$R, Q \in \mathcal F \cup \mathcal F'$ as
  long as~$Q \neq R'$ and~$R \neq Q'$. The only pairs we have not yet shown to
  be resolved by~$S$ are of the form $R, R'$ for~$R \in \mathcal F$ with its
  copy~$R' \in \mathcal F'$. Since~$H \subseteq S$ is a hitting set for
  $(U, \mathcal F)$, every set~$R \in \mathcal F$ is adjacent to at least one
  vertex in~$H$ while~$R'$ has no neighbours at all in~$U$. Thus all such pairs
  are resolved by~$S$ and we conclude that~$S$ is a resolving set.

  In the other direction, assume that~$S$ is a resolving set of size~$k$
  for~$G$. Since for each pair of twins at least one vertex has to be in any
  resolving set, we may assume, without loss of generality, that~$I_U \cup
  I_{\mathcal F} \cup \{a_U, a, a_\mathcal F\} \subseteq S$. Let us call this
  collection of~$k - \ell$ vertices~$S' \subseteq S$ and let us see what it
  resolves in~$G$. As argued above, every pair except those of the form~$R \in
  \mathcal F$, $R' \in \mathcal F'$ are certainly resolved. We first need to
  argue that~$S'$ indeed does not resolve those pairs: this is immediately
  obvious for landmarks in~$I_{\mathcal F} \cup \{a_U, a, a_\mathcal F\}$
  since~$R, R'$ share the same neighbours inside this set. For landmarks
  in~$I_U$, note that all vertices in $\mathcal F \cup \mathcal F'$ are at
  exactly distance two from every vertex in~$I_U$ via the apex vertex~$a$
  (or~$a'$). Hence~$S'$ cannot resolve any pair~$R, R' \in \mathcal F \cup
  \mathcal F'$ and these pairs must then be resolved by the remaining~$\ell$
  vertices in~$S\setminus S'$. All vertices outside of~$U \cup \mathcal F \cup
  \mathcal F'$ are either selected or twins to selected vertices, hence we may
  assume that $S\setminus S' \subseteq U \cup \mathcal F \cup \mathcal F'$.

  First, consider a potential landmark~$ R \in \mathcal F$. Since~$R$ has
  distance exactly two to every vertex in~$\mathcal F \cup \mathcal F'$ except
  itself, such a selection would only distinguish~$R$ from all other vertices
  and not resolve any other pair. Thus we can as well choose any vertex
  in~$N(R) \cap U$ instead and potentially resolve more pairs, thus we may
  assume that~$S \cap \mathcal F = \emptyset$. 

  Let us split~$S \setminus S'$ into~$S_U := U \cap (S \setminus S')$
  and~$S_{\mathcal F'} := \mathcal F' \cap (S \setminus S')$. Again,
  $S_{\mathcal F'}$ only distinguishes~$S_{\mathcal F'}$ from the rest
  of~$\mathcal F \cup \mathcal F'$. Thus~$S_U$ necessarily distinguishes all
  pairs~$R, R'$ with~$R' \not \in S_{\mathcal F'}$ and therefore~$S_U$ hits
  all sets~$R \in \mathcal F$ for which~$R' \not \in S_{\mathcal F'}$. We
  finally construct a hitting set~$H$ of size~$\ell$ as follows: we take all
  vertices in~$S_U$ and for each pair~$R,R' \in \mathcal F \cup \mathcal F'$
  with~$R' \in S_{\mathcal F'}$ we select one (arbitrary) neighbour~$N(R) \cap U$.
  By the previous observation, $H$ is a hitting set for~$\mathcal F$ of size~$\ell$
  and we conclude that~$(U,\mathcal F,\ell)$ is a yes-instance.

  This concludes the parameter preserving transformation. Let us conclude by
  checking that the parameter $|X| + k$ is polynomial in $n$ and $\ell$: 
  \begin{align*}
    |X| + k &= (2t_U  + 2t_{\mathcal F} + n + 6) + (\ell + t_U + t_{\mathcal F} + 3) \\
            &= 3t_U + 3t_{\mathcal F} + n + \ell + 9 \\
            &= O( \log m + \log n + n + \ell ) = O(n + \ell),
  \end{align*}
  where we used that~$m \leq 2^n$.
\end{proof}

\noindent
We note that this reduction also gives evidence against a more general
form of kernelization. Where a standard kernel can be understood as a
many-one reduction from a problem to itself, with output size bounded
by a function of the parameter, a \emph{Turing kernel} is the
corresponding Turing reduction notion.  In other words, informally, a
Turing kernel is a polynomial-time procedure that solves a
parameterized problem, with access to an oracle for the problem but
with a bound $f(k)$ on the maximum length of the questions it may ask
of the oracle. A \emph{polynomial Turing kernel} is a Turing kernel 
with a bound $f(k)=k^{O(1)}$ on the question size. For a more formal
definition, see~\cite{HermelinKSWW15,CyganFKLMPPS15}. 
It is known that there are parameterized problems that do not allow a
polynomial kernel unless the polynomial hierarchy collapses, but which
do allow polynomial Turing kernels; cf.~\cite{HermelinKSWW15,ThomasseTV17,JansenM15}.

Although we do not have a framework for excluding polynomial Turing
kernels that is as powerful as that for excluding standard polynomial
kernels, Hermelin \etal~\cite{HermelinKSWW15} defined a hierarchy of
complexity classes, conjectured to represent problems that do not
allow polynomial Turing kernels.  The most basic and most common of
these hardness classes is WK[1], which is in turn contained in a
larger class MK[2]. It is conjectured in~\cite{HermelinKSWW15} that no
WK[1]-hard problem has a polynomial Turing kernel. 
Since \Problem{Hitting Set[$n$]} is known to be
MK[2]-hard~\cite{HermelinKSWW15}, the above reduction gives the
following. 

\begin{corollary}
  \Problem{Metric Dimension[$\text{VC}+k$]} is MK[2]-hard (hence also
  WK[1]-hard) under polynomial parameter transformations, and does not
  allow a polynomial Turing kernel unless \Problem{CNF-SAT$[n]$} and
  every other problem in MK[2] does.
\end{corollary}

\section{Conclusion} 

\noindent
We initiated the study of the parameterized complexity of
the dual of the classic \Problem{Metric Dimension} problem and obtained a
polynomial kernel as well as a single-exponential $\FPT$ algorithm. To the best
of our knowledge, this is the first non-trivial parameterization for
\Problem{Metric Dimension} which leads to a polynomial kernel. Since our focus
in this article was on obtaining new classification results, we leave the
improvement of the kernel size or a potential proof of a lower bound on the
bitsize of our kernel, to future work.

In addition, we note that it remains open whether \Problem{Metric Dimension}
is polynomial time solvable even on series-parallel graphs. Since series-
parallel graphs are precisely the graphs of treewidth 2, a negative answer
would also imply that there is no XP algorithm for \Problem{Metric Dimension}
parameterized by the treewidth. Consequently, a natural starting point of
enquiry towards addressing this question could be  the study of the
parameterized complexity of \Problem{Metric Dimension} parameterized by
treewidth. 

\paragraph{Acknowledgement} Gutin was partially supported by Royal Society Wolfson Research Merit Award.


\begin{thebibliography}{10}

\bibitem{BasavarajuFRS16}
M.~Basavaraju, M.~C. Francis, M.~S. Ramanujan, and S.~Saurabh.
\newblock Partially polynomial kernels for {Set Cover} and {Test Cover}.
\newblock {\em {SIAM} J. Discrete Math.}, 30(3):1401--1423, 2016.

\bibitem{BelmonteFGR17}
R.~Belmonte, F.~V. Fomin, P.~A. Golovach, and M.~S. Ramanujan.
\newblock Metric dimension of bounded tree-length graphs.
\newblock {\em {SIAM} J. Discrete Math.}, 31(2):1217--1243, 2017.

\bibitem{ChartrandEJO00}
G.~Chartrand, L.~Eroh, M.~A. Johnson, and O.~Oellermann.
\newblock Resolvability in graphs and the metric dimension of a graph.
\newblock {\em Discrete Applied Mathematics}, 105(1-3):99--113, 2000.

\bibitem{CrowstonGJSY12}
R.~Crowston, G.~Gutin, M.~Jones, S.~Saurabh, and A.~Yeo.
\newblock Parameterized study of the {Test Cover Problem}.
\newblock In {\em Mathematical Foundations of Computer Science 2012 - 37th
  International Symposium, {MFCS} 2012, Bratislava, Slovakia, August 27-31,
  2012. Proceedings}, pages 283--295, 2012.

\bibitem{CyganFKLMPPS15}
M.~Cygan, F.~V. Fomin, L.~Kowalik, D.~Lokshtanov, D.~Marx, M.~Pilipczuk,
  M.~Pilipczuk, and S.~Saurabh.
\newblock {\em Parameterized Algorithms}.
\newblock Springer, 2015.

\bibitem{DiazPSL12}
J.~D{\'{\i}}az, O.~Pottonen, M.~J. Serna, and E.~J. van Leeuwen.
\newblock On the complexity of metric dimension.
\newblock In {\em {ESA} 2012}, volume 7501 of {\em Lecture Notes in Computer
  Science}, pages 419--430. Springer, 2012.

\bibitem{ColsAndIds}
M.~Dom, D.~Lokshtanov, and S.~Saurabh.
\newblock Kernelization lower bounds through colors and {ID}s.
\newblock {\em {ACM} Trans. Algorithms}, 11(2):13:1--13:20, 2014.

\bibitem{DowneyF13}
R.~G. Downey and M.~R. Fellows.
\newblock {\em Fundamentals of Parameterized Complexity}.
\newblock Texts in Computer Science. Springer, 2013.

\bibitem{Eppstein15}
D.~Eppstein.
\newblock Metric dimension parameterized by max leaf number.
\newblock {\em J. Graph Algorithms Appl.}, 19(1):313--323, 2015.

\bibitem{EpsteinLW12}
L.~Epstein, A.~Levin, and G.~J. Woeginger.
\newblock The (weighted) metric dimension of graphs: Hard and easy cases.
\newblock In {\em {WG} 2012}, volume 7551 of {\em Lecture Notes in Computer
  Science}, pages 114--125. Springer, 2012.

\bibitem{FominLSZ}
F.~V. Fomin, D.~Lokshtanov, S.~Saurabh, and M.~Zehavi.
\newblock {\em Kernelization: Theory of Parameterized Preprocessing}.
\newblock Springer, in preparation.

\bibitem{FoucaudMNPV17}
F.~Foucaud, G.~B. Mertzios, R.~Naserasr, A.~Parreau, and P.~Valicov.
\newblock Identification, location-domination and metric dimension on interval
  and permutation graphs. {II.} algorithms and complexity.
\newblock {\em Algorithmica}, 78(3):914--944, 2017.

\bibitem{SparseKernels}
J.~Gajarsk{\'{y}}, P.~Hlinen{\'{y}}, J.~Obdrz{\'{a}}lek, S.~Ordyniak, F.~Reidl,
  P.~Rossmanith, F.~S. Villaamil, and S.~Sikdar.
\newblock Kernelization using structural parameters on sparse graph classes.
\newblock {\em J. Comput. Syst. Sci.}, 84:219--242, 2017.

\bibitem{GareyJ79}
M.~R. Garey and D.~S. Johnson.
\newblock {\em Computers and Intractability: {A} Guide to the Theory of
  NP-Completeness}.
\newblock W. H. Freeman, 1979.

\bibitem{GutinJY11}
G.~Gutin, M.~Jones, and A.~Yeo.
\newblock Kernels for below-upper-bound parameterizations of the hitting set
  and directed dominating set problems.
\newblock {\em Theor. Comput. Sci.}, 412(41):5744--5751, 2011.

\bibitem{HararyM76}
F.~Harary and R.~A. Melter.
\newblock On the metric dimension of a graph.
\newblock {\em Ars Combinatoria}, 2:191--195, 1976.

\bibitem{HartungN13}
S.~Hartung and A.~Nichterlein.
\newblock On the parameterized and approximation hardness of metric dimension.
\newblock In {\em Proceedings of the 28th Conference on Computational
  Complexity, {CCC} 2013, K.lo Alto, California, USA, 5-7 June, 2013}, pages
  266--276, 2013.

\bibitem{HermelinKSWW15}
D.~Hermelin, S.~Kratsch, K.~Soltys, M.~Wahlstr{\"{o}}m, and X.~Wu.
\newblock A completeness theory for polynomial ({Turing}) kernelization.
\newblock {\em Algorithmica}, 71(3):702--730, 2015.

\bibitem{HoffmannW12}
S.~Hoffmann and E.~Wanke.
\newblock Metric dimension for {Gabriel} unit disk graphs is {NP}-complete.
\newblock In {\em {ALGOSENSORS} 2012}, volume 7718 of {\em Lecture Notes in
  Computer Science}, pages 90--92. Springer, 2012.

\bibitem{JansenM15}
B.~M.~P. Jansen and D.~Marx.
\newblock Characterizing the easy-to-find subgraphs from the viewpoint of
  polynomial-time algorithms, kernels, and {Turing} kernels.
\newblock In {\em {SODA}}, pages 616--629. {SIAM}, 2015.

\bibitem{KhullerRR96}
S.~Khuller, B.~Raghavachari, and A.~Rosenfeld.
\newblock Landmarks in graphs.
\newblock {\em Discrete Applied Mathematics}, 70(3):217--229, 1996.

\bibitem{Kratsch14}
S.~Kratsch.
\newblock Recent developments in kernelization: {A} survey.
\newblock {\em Bulletin of the {EATCS}}, 113, 2014.

\bibitem{LokshtanovMS12}
D.~Lokshtanov, N.~Misra, and S.~Saurabh.
\newblock Kernelization - preprocessing with a guarantee.
\newblock In {\em The Multivariate Algorithmic Revolution and Beyond - Essays
  Dedicated to Michael R. Fellows on the Occasion of His 60th Birthday}, pages
  129--161, 2012.

\bibitem{RazgonO09}
I.~Razgon and B.~O'Sullivan.
\newblock Almost 2-{SAT} is fixed-parameter tractable.
\newblock {\em J. Comput. Syst. Sci.}, 75(8):435--450, 2009.

\bibitem{Slater75}
P.~J. Slater.
\newblock Leaves of trees.
\newblock In {\em Proceedings of the {S}ixth {S}outheastern {C}onference on
  {C}ombinatorics, {G}raph {T}heory, and {C}omputing ({F}lorida {A}tlantic
  {U}niv., {B}oca {R}aton, {F}la., 1975)}, pages 549--559. Congressus
  Numerantium, No. XIV. Utilitas Math., Winnipeg, Man., 1975.

\bibitem{ThomasseTV17}
S.~Thomass{\'{e}}, N.~Trotignon, and K.~Vuskovic.
\newblock A polynomial {Turing}-kernel for weighted independent set in
  bull-free graphs.
\newblock {\em Algorithmica}, 77(3):619--641, 2017.

\end{thebibliography}
\end{document}